\documentclass[preprint,12pt]{elsarticle}

\usepackage{graphicx}
\usepackage{amssymb}

\usepackage{mathtools}
\usepackage{cases}
\usepackage{amsfonts}
\usepackage{amssymb}
\usepackage{twcal}
\usepackage{graphicx,url,subfig}
\usepackage{algorithm}
\usepackage{algpseudocode}
\usepackage{algorithmicx}
\usepackage{hyperref}
\usepackage{xcolor}
\usepackage{pstricks}

\usepackage{amssymb,amsmath,amsthm}

\newtheorem{theorem}{Theorem}
\newtheorem{lemma}{Lemma}
\newtheorem{corollary}{Corollary}
\newtheorem{observation}{Observation}

\newlength\myindent
\newcounter{cases}
\newcounter{subcases}
\newenvironment{mycases}
  {%
    \setcounter{cases}{0}%
    \def\case
      {%
        \par\noindent
        \refstepcounter{cases}%
        \textbf{Case \thecases.}
      }%
  }
  {%
    \par
  }

\renewcommand*\thecases{\arabic{cases}}

\journal{DAM}

\begin{document}

\begin{frontmatter}


\title{The $k$-Colorable Unit Disk Cover Problem}


\author{Monith S. Reyunuru}

\address{Amazon Development Center, Hyderabad, India}

\author{Kriti Jethlia}
\author{Manjanna Basappa\corref{mycorrespondingauthor}}

\address{Birla Institute of Technology \& Science Pilani, Hyderabad Campus, India}

\cortext[mycorrespondingauthor]{Corresponding author}
\ead{manjanna@hyderabad.bits-pilani.ac.in}

\begin{abstract}
In this article, we consider colorable variations of the Unit Disk Cover ({\it UDC}) problem as follows.

{\it $k$-Colorable Discrete Unit Disk Cover ({\it $k$-CDUDC})}: Given a set $P$
of $n$ points, and a set $D$ of $m$ unit disks (of radius=1), both lying in the plane, and a parameter $k$, the
objective is to compute a set $D'\subseteq D$ such that every point in $P$ is covered by at least one disk in $D'$ and there exists a function $\chi:D'\rightarrow C$ that assigns colors to disks in $D'$ such that for any $d$ and $d'$ in $D'$ if $d\cap d'\neq\emptyset$, then $\chi(d)\neq\chi(d')$, where $C$ denotes a set containing $k$ distinct colors.

For the {\it $k$-CDUDC} problem, our proposed algorithms approximate the number of colors used in the coloring if there exists a $k$-colorable cover. We first propose a 4-approximation algorithm in $O(m^{7k}n\log k)$
time for this problem and then show that the running time can be improved by a multiplicative factor of $m^k$, where a positive integer $k$ denotes the cardinality of a color-set. The previous best known result for the problem when $k=3$ is due to the recent work of Biedl et al., (2021)\cite{BBL19}, who proposed a 2-approximation algorithm in $O(m^{25}n)$ time. For $k=3$, our algorithm runs in $O(m^{18}n)$ time, faster than the previous best algorithm, but gives a 4-approximate result. We then generalize our approach to yield a family of $\rho$-approximation algorithms in $O(m^{\alpha k}n\log k)$ time, where $(\rho,\alpha)\in \{(4, 7), (6,5), (7, 5), (9,4)\}$. We further generalize this to exhibit a $O(\frac{1}{\tau})$-approximation algorithm in $O(m^{\alpha k}n\log k)$ time for a given $1 \leq \tau \leq 2$, where $\alpha=O(\tau^2)$. We also extend our algorithm to solve the {\it $k$-Colorable Line Segment Disk Cover ({\it $k$-CLSDC})} and {\it $k$-Colorable Rectangular Region Cover ({\it $k$-CRRC})} problems, in which instead of the set $P$ of $n$ points, we are given a set $S$ of $n$ line segments, and a rectangular region $\cal R$, respectively.

\end{abstract}


\end{frontmatter}



\section{Introduction}
Our motivation for studying the problem arises from practical applications in the frequency/channel assignment problem in wireless/cellular networks. In ad-hoc mobile networks, each host(station/tower) is equipped with a Radio-Frequency (RF) transceiver to provide reliable transmission inside a circular range, represented by a disk, within some distance. Each wireless client is equipped with corresponding receivers. The clients themselves are represented by a set of points $P$ in a plane. The disks representing the range (which is presumably the same for all stations) of each potential host is represented by the set $D$. In the spirit of reducing interference in broadcast and other energy-saving measures, we aim to limit or reduce the number of different frequencies(channels) assigned to each, represented by coloring. Typically, (Wi-Fi) networks are built with 3 independent channels \cite{BHLL10}, hence the motivation for a study on the {\it 3-CDUDC} problem. In the same spirit, we generalize the {\it 3-CDUDC} to the {\it $k$-CDUDC} problem, where $k>0$ is an integer. We further generalize the problem by considering line segments and a continuous rectangular region as representing potential wireless clients (resp. the {\it $k$-CLSDC} and {\it $k$-CRRC} problems), instead of points.

\subsection{Related Work}

The {\it 3-CDUDC} problem, to the best of our knowledge, was first studied by Biedl et al., \cite{BBL19}. They gave a 2-approximation algorithm in $O(nm^{25})$ time for the {\it 3-CDUDC} problem. Their approach first partitions the plane into horizontal strips, solves the problem for every strip optimally, then returns the union of solutions of all strips. To solve the problem for any strip they show that at most a constant number of disks of an optimal solution intersect any vertical line. Based on this, they define a directed acyclic graph such that there exists a path from source to a destination corresponding to this optimal solution. In this paper, we attempt to improve upon this impractical $O(nm^{25})$ running time. Our approach, however, focuses on the specific geometric properties that arise from the dual conditionals of the problem statement. Although both of the approaches, initially, begin by dividing the plane, we recognize a unique bound that exists in our need to bound the colorability and provide a novel solution in the same regard. 

A notion of {\it conflict-free coloring (CF-coloring)} was introduced by Even et al., \cite{ELRS03}.
and Smorodinsky \cite{S03}.
In the {\it CF-coloring} problem we are given a set of points (representing client locations) and a set of base stations, the objective is to assign colors (representing frequencies) to the base stations such that any client lying within the range of at least one base station is covered by the base station whose color is different from the colors of the other base stations covering the client, and the number of colors used should be as minimum as possible. Here, the range of base stations is modeled as regions e.g., disks or other geometric objects. Even et al., \cite{ELRS03} proved that $O(\log n)$ colors are always sufficient to {\it CF-color} a set of disks in the plane, and in the worst case, $\Omega(\log n)$ colors are required. Note that this {\it CF-coloring} of disks is different from our notion of $k$-colorable disk cover of points. In the former overlapping disks may be given the same color if they dont share a client, whereas in the {\it $k$-CDUDC} overlapping disks must be colored with distinct colors regardless of whether they cover a common client. A generalization of {\it CF-coloring} is called a {\it $k$-fault-tolerant CF-coloring}. Cheilaris et al., \cite{CGRS14} presented a polynomial-time $(5-\frac{2}{k})$-approximation algorithm for the {\it $k$-fault-tolerant CF-coloring} in 1-dimensional space. Horev et al., \cite{HKS10} proved that $O(k\log n)$ colors are sufficient for any set of $n$ disks in the plane.  For {\it dynamic CF-coloring} and results on {\it CF-coloring} of other geometric objects, we refer to \cite{BM19} and references therein. 

A related problem of the {\it $k$-CDUDC} problem in the literature is the {\it Discrete Unit Disk Cover} ({\it DUDC}) problem. In the {\it DUDC} problem, we are given a set $P$ of $n$ points and a set $D$ of $m$ unit disks, our goal is to select as the smallest number of disks from $D$ as possible such that the union of these selected disks covers all points in $P$. As in the {\it $k$-CDUDC}, here also, the sets $P$ and $D$ can be considered as representing a set of wireless clients and a set of base stations or towers, respectively. The {\it DUDC} problem is \texttt{NP}-hard and is a very well studied one. There is a polynomial time approximation scheme (PTAS) with impractical running time for this problem \cite{MR10}. The current best approximation algorithm with reasonable running time is $(9+\epsilon)$ for any $\epsilon >0$ \cite{BAD15}. However, a series of approximation algorithms have been proposed for this problem by various authors over the past two decades, and a complete survey on this can be found in \cite{FL12}. When a line segment is used to represent a potential wireless client, the {\it DUDC} problem becomes a {\it Line Segment Disk Cover} ({\it LSDC}) problem. In a similar line, there is another variant of the {\it DUDC} problem, a {\it Rectangular Region Cover (RRC)} problem, in which all the continuous set of points lying in a rectangular region represent wireless clients.  All the available results for the {\it DUDC} problem also extend to the {\it LSDC} and {\it RRC} problems \cite{B18}, with slightly different running time.
We also extend our results for the {\it $k$-CDUDC} problem to solve the colorable variants of the {\it LSDC} and {\it RRC} problems, namely, the {\it $k$-CLSDC} and {\it $k$-CRRC} problems.  

\section{$k$-CDUDC Problem}\label{section-2.2}

In this section we consider the following problem.
\begin{itemize}
 \item[] {\it $k$-Colorable Discrete Unit Disk Cover ({\it $k$-CDUDC})}: Given a set $P$
of $n$ points, and a set $D$ of $m$ unit disks (of radius=1), both lying in the plane, and a parameter $k$, the
objective is to compute a set $D'\subseteq D$ that covers all points in $P$ such that the set $D'$ can be partitioned into $\{ D_1', D_2', \ldots,
D_k'\}$, where for each $a \in \{1, 2, \ldots, k\}$ the disks in $D_a'$ are pairwise disjoint, i.e., the disks in $D'$ can be colored with at most $k$ colors such that
the overlapping disks receive distinct colors and every point in $P$ is covered by a disk in $D'$.
\end{itemize}

As it was pointed out in \cite{BBL19} that there is a related problem, namely, {\it Unit Disk Chromatic Number (UDCN)} problem, that aims to color all nodes in a given unit disk graph with at most $k$ colors. The {\it UDCN} problem is \texttt{NP}-hard for any $k\geq 3$ \cite{CCJ90}. Similar to Biedl et al. \cite{BBL19}, we can center a set $D$ of $m$ unit disks in the plane such that there are at least $k+1$ pairwise non-disjoint disks that have a common intersection region and a unit disk graph $G_{D}=(V_{D}, E_{D})$ induced by $D$ is connected. Let us then place a set $P$ of $n$ points in this intersection region. Now observe that the set $P$ has a cover which is at most $k$-colorable, whereas the graph $G_{D}$ is at least $(k+1)$-colorable. Hence, the {\it $k$-CDUDC} problem is different from the {\it UDCN} problem. Biedl et al. \cite{BBL19} showed that the {\it 3-CDUDC} problem is \texttt{NP}-hard by carefully incorporating a set $P$ of $n$ points in the \texttt{NP}-hard proof of the {\it UDCN} problem with $k=3$ in \cite{CCJ90}. This directly implies that the $k$-{\it CDUDC} is \texttt{NP}-hard since the $k$-{\it CDUDC} is a generalization of {\it 3-CDUDC}. It is also easy to see that the $k$-{\it CDUDC} problem belongs to the class \texttt{NP}, as follows: Here, the certificate for any Yes instance of $k$-{\it CDUDC} is a set of $k$ distinct colors identified by non-negative integers $1,2, \ldots,k$, and a mapping $\chi:D'\rightarrow \{1,2, \ldots,k\}$, where $D'\subseteq D$. A polynomial time verifier checks if every point in $P$ is covered by a disk in $D'$ and for every pair of disks $d,d'\in D'$ if $d\cap d'\neq \emptyset$, whether it is the case that $\chi(d)\neq \chi(d')$.

\subsection{4-Approximate Algorithm}\label{section2.1}

Here, our algorithm is based on partitioning the plane containing points into a grid and then determining bound on the number of unit disks that can participate in any $k$-colorable covering of points lying within any square of the grid. We first define a grid of width $\tau$ units that partitions the plane into squared regions. Each of these squared regions is a grid cell with a size $\tau\times\tau$. For simplicity assume no point of $P$ lies on the boundary of these grid cells. Let us associate a unique ID $id_{\cal C}$ to each grid cell ${\cal C}$ as follows; let $p=(x_p, y_p)$ be a point in ${\cal C}$ and $\tau$ be the grid width, then $id_{\cal C}=(\lfloor\frac{x_p}{\tau}\rfloor, \lfloor\frac{y_p}{\tau}\rfloor)$, (see Fig. \ref{fig-2}). Note that each grid cell has a unique ID associated with it but multiple points can be associated with the same ID (if they lie within the corresponding grid cell). Let $id_{{\cal C}_1}$ and $id_{{\cal C}_2}$ be any two arbitrary grid cells with base points $(x_1,y_1)$ and $(x_2,y_2)$ respectively. We define the greater than operator for an ID as follows: $(id_{{\cal C}_1} > id_{{\cal C}_2}) \iff ((x_1 > x_2) \land (y_1 > y_2)) \lor  ((x_1 = x_2) \land (y_1 > y_2)) \lor ((x_1 > x_2) \land (y_1 = y_2))$. Note that our defintion of $id_{\cal C}$ implies that the grid cells are indexed from bottom-left to top-right, this operator simply indicates the order of iteration that is followed by our algorithm. We move left to right row-wise starting from the bottom-left corner cell. A pre-defined order is essential to our handing-over logic at line 11 of Algorithm \ref{4.alg:SSC}.

\begin{figure}[h]
\centering
\includegraphics[scale=0.15]{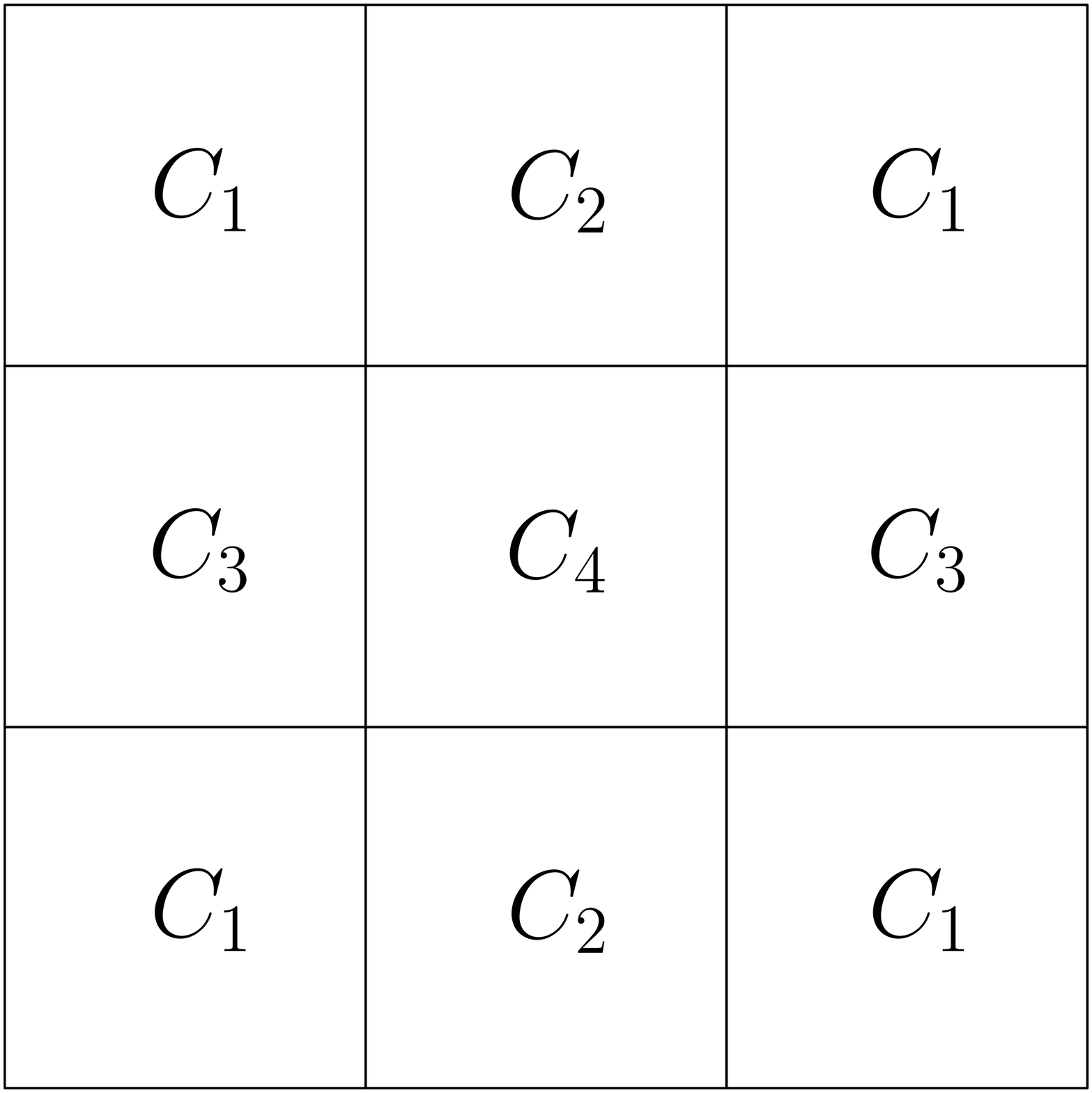}
\begin{picture}(0,0)

  \put(  -112, 2){$[0,0]$}
    \put(  -112, 39){$[0,1]$}
 \put(  -112, 76){$[0,2]$}
   \put(  -76, 2){$[1,0]$}
    \put(  -76, 39){$[1,1]$}
 \put(  -76, 76){$[1,2]$}
   \put(  -39, 2){$[2,0]$}
    \put(  -39, 39){$[2,1]$}
 \put(  -39, 76){$[2,2]$}
\end{picture}
\caption {Assignment of unique ID's $id_{\cal C}$ and color sets for a grid ${\cal G}$ with $\tau=2$}
\label{fig-2}
\end{figure}

Given a grid cell ${\cal C}'$, the following lemma provides a bound on the cardinality of the $k$-colorable unit disks covering all points lying within the grid cell ${\cal C}'$. Let $D_{{\cal C}'}\subseteq D$ be the set of $k$-colorable unit disks covering all points of $P$ lying within the grid cell ${\cal C}'$. The proof of the lemma is based on the observation that determining this bound is the same as determining a maximum number of disjoint unit disks that could potentially intersect ${\cal C}'$.

\begin{observation}\label{observation1.1}
 If ${\cal C}'$ is a grid cell of size $\tau\times \tau$, then the maximum number of pairwise disjoint unit disks that could potentially intersect ${\cal C}'$ is at most $2\tau+2+(\frac{\tau}{2})^2$ if $\tau$ is even, and is atmost $4\times\lceil\frac{\tau}{2}\rceil+4+(\frac{\tau}{2})^2$ if $\tau$ is odd.
\end{observation}

\begin{proof}
 We will provide an upper bound to the number of pairwise disjoint unit disks that can cover a square ${\cal C}'$ of side length $\tau$. Let us prove this by considering the two cases: $\tau$ being even and odd. 

Since we aim at bringing an upper bound to the number of disks that have a common intersection point with ${\cal C}'$, we divde the region of ${\cal C}'$ into two parts; the inner part of the square and the union of its outer edges on which these common intersection points can lie. To maximize the number of disks, it is intuitive to keep them as far as possible to increase the spacing between disks and thereby trying to increase the number of disks.

 \noindent{\bf ${\cal C}'$ with even side length:} When $\tau$ is a multiple of 2, it is quite intutive that a symmetric pattern is likely to give the best results. So we attempt two types of symmetric pattern.

\textbf{Case 1:} Considering the square ${\cal C}'$ to be symmetric along the vertical axis, we arrange the disks in two possible cases: either a disk is arranged with edge of ${\cal C}'$ as tangent such that the center 
of ${\cal C}'$ lies vertically above/below the disk (see Fig. \ref{bound1a}), or the vertical partition is tangent to some of the disks (see Fig. \ref{bound1b}). In the first case the maximum number of disks along horizontal part of outer edge would be $2\times\frac{\tau}{2}$, since the diameter of disk is 2. In the second case it can be shown that the maximum number of disks along the horizontal part would be $2\times(\frac{\tau}{2} + 1)$.

\begin{figure}
\centering
\subfloat[Center of disk along the vertical axis\label{bound1a}]{
\includegraphics[width=.45\textwidth]{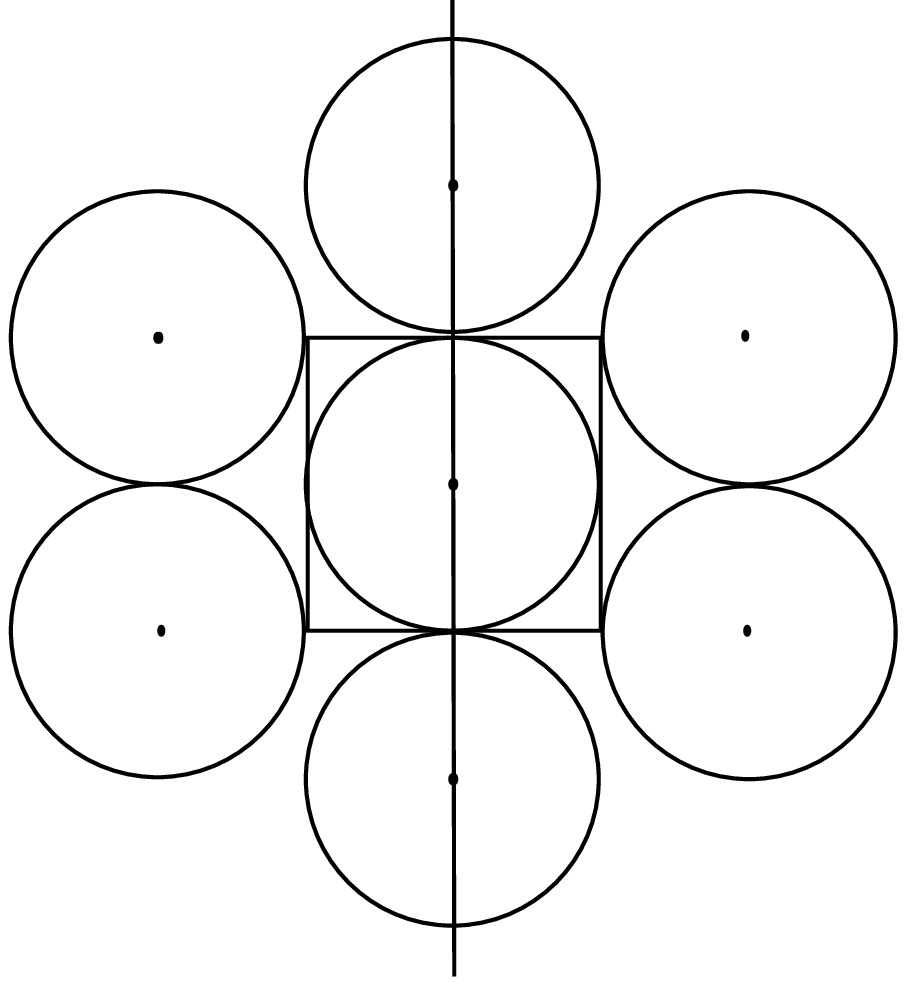}
}
\subfloat[Vertical line dividing the square as tangent\label{bound1b}]{
\includegraphics[width=.45\textwidth]{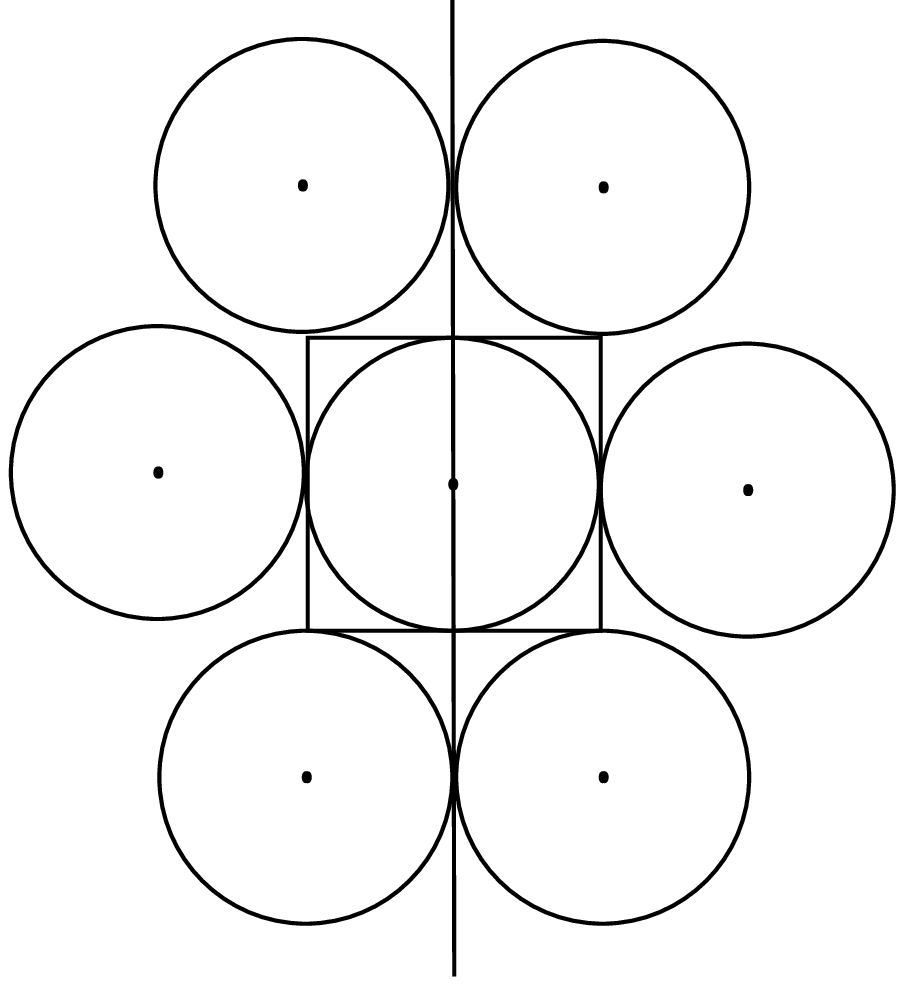}
}
\caption{Possible symmetric arrangement along vertical axis for even values of $\tau$.}
\label{1.observ}
\end{figure}

However, as we see for the case $\tau$=2 (shown in Fig \ref{bound1a} and Fig \ref{bound1b}), both the cases break down to the same case since if for a given pair of parallel edges of ${\cal C}'$ if one of the above case is true, then for the other pair the other case stands true. By induction we can prove that this stands true for all the even values of $\tau > 2$. Through this we get an upper limit in the number of disks along the edges. For the inner area of ${\cal C}'$ the maximum number of disjoint disks are $(\frac{\tau}{2})^2$.
So the maximum number of disjoint disks are ($\tau$/2)$^2$ +2 $\times$ $\tau$/2 +2 $\times$ ($\tau$/2 + 1), which is 2$\tau$+2+($\tau$/2)$^2$.
 

\textbf{Case 2:} Considering ${\cal C}'$ to be diagonally symmetric, here we consider one of its diagonals as the symmetry line and place the disk such that their center is on the diagonal and intersects the square at its corner. As the disk covers a part of the edge of the square, say $\delta$  where $\delta < 1$ (i.e., radius of disk) and  $\delta > 0$, apart from the case where the disks are arranged at the diagonals again (see Fig. \ref{bound1c}), the maximum number of disks on an edge is still equivalent to $\frac{\tau}{2}$. The number of disks within the interiors of ${\cal C}'$ as stated above in the previous case will remain same. Therefore, the total number of maximum possible pairwise disjoint disks are $2+4{\times}(\frac{\tau}{2})+(\frac{\tau}{2})^2$, which is again equivalent to $2\tau+2+(\frac{\tau}{2})^2$ (see Fig. \ref{bound1c}, Fig \ref{bound1d}, Fig \ref{bound1e}, Fig \ref{bound1f}, and Fig \ref{bound1g} for illustration of the case ${\tau} = 2$).

\begin{figure}
\centering
\subfloat[Disks arranged along one of the diagonal\label{bound1c}]{
\includegraphics[width=.45\textwidth]{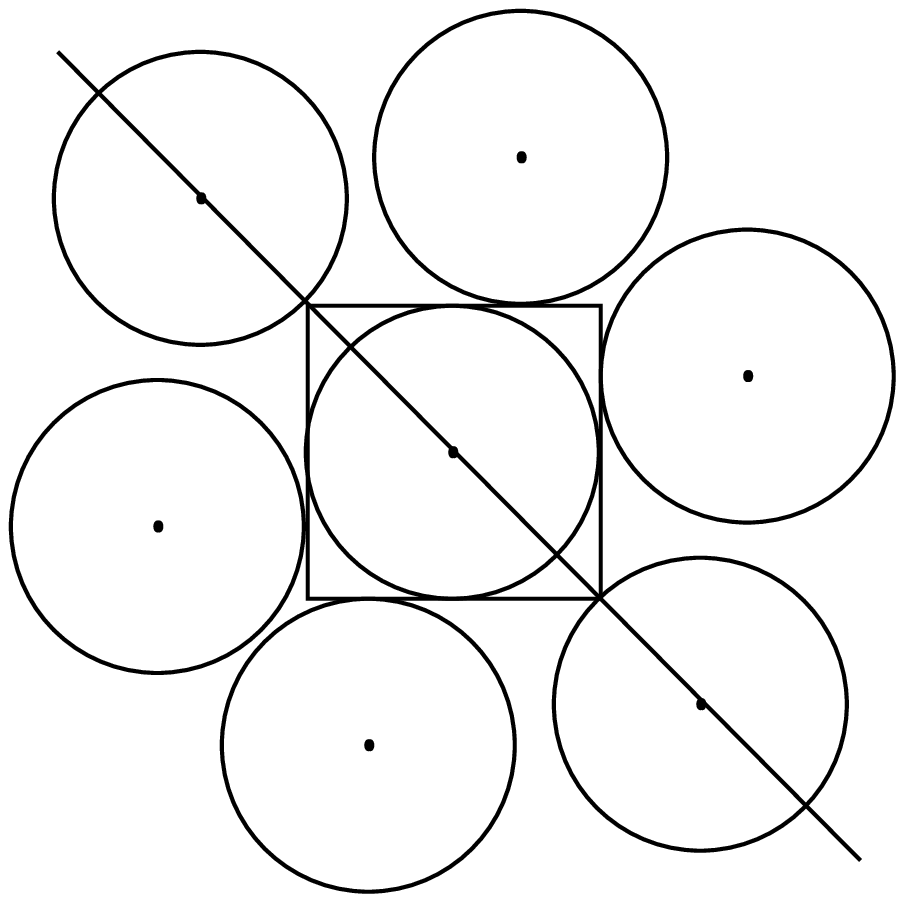}
}
\subfloat[Disks arranged are along both the diagonals\label{bound1d}]{
\includegraphics[width=.45\textwidth]{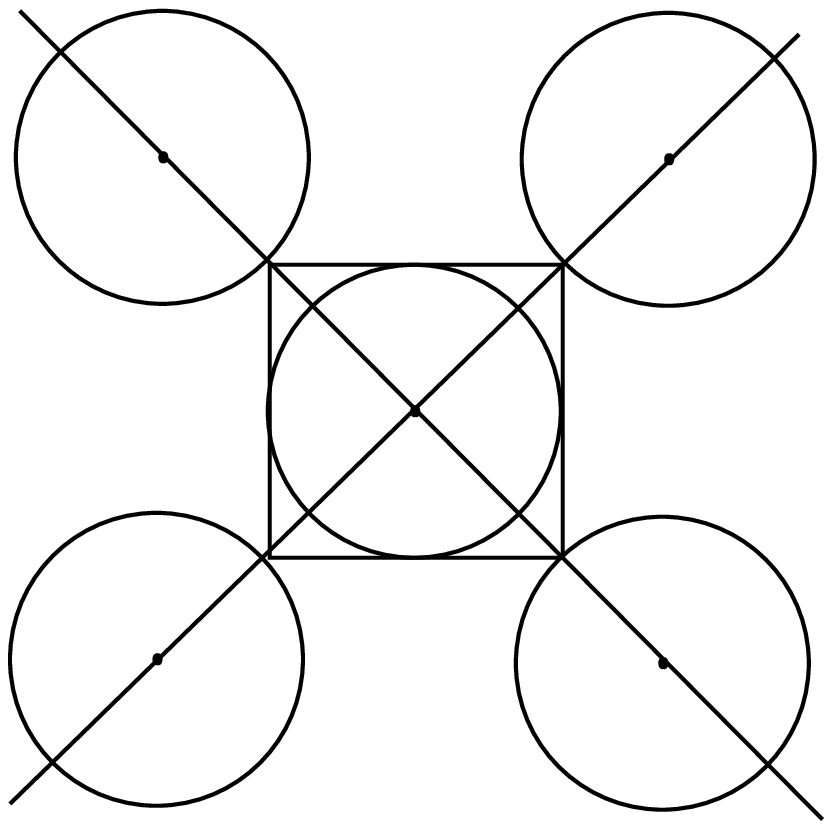}
}
\caption{Possible symmetric arrangement along diagonal for even values of $\tau$.}
\label{1.observb}
\end{figure}

 \begin{figure}
\centering
\subfloat[Complete coverage as much closely as possible\label{bound1e}]{
\includegraphics[width=.45\textwidth]{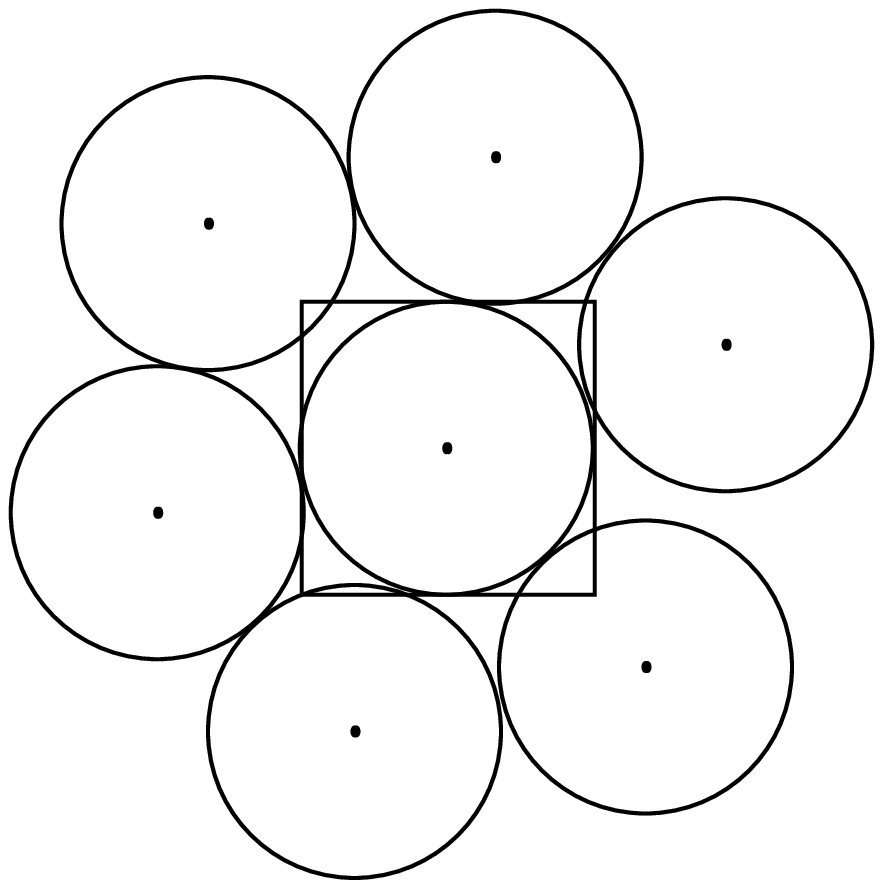}
}
\subfloat[Symmetric coverage\label{bound1f}]{
\includegraphics[width=.45\textwidth]{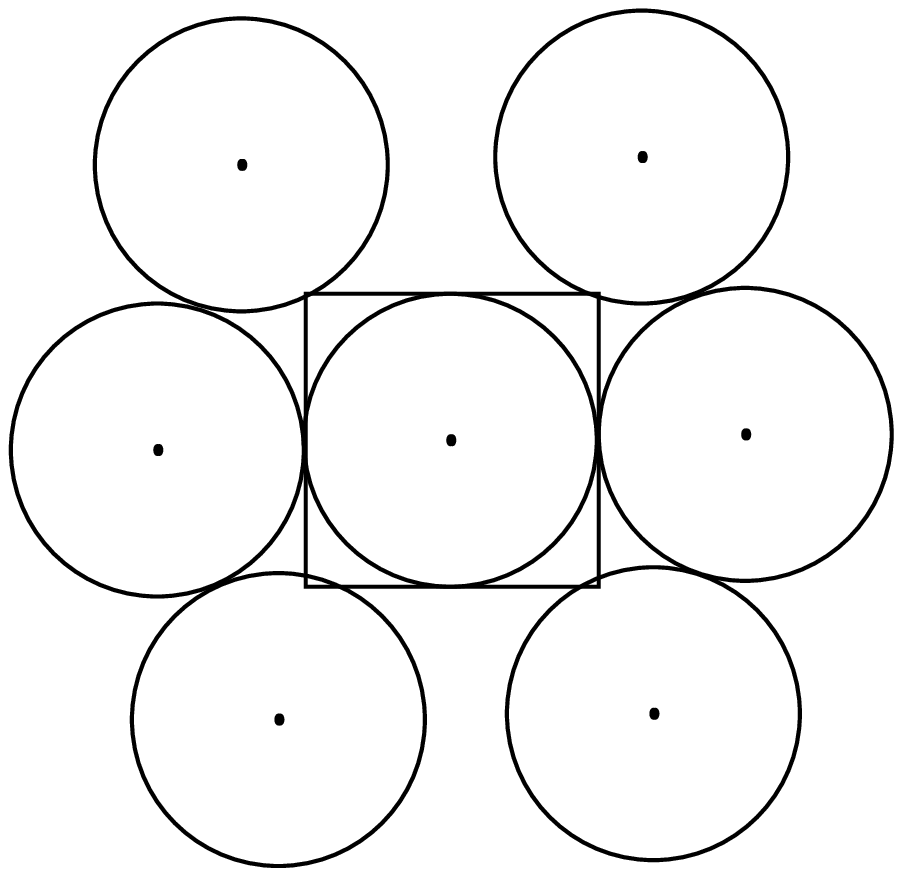}
}

\subfloat[Asymmetric coverage\label{bound1g}]{
\includegraphics[width=.45\textwidth]{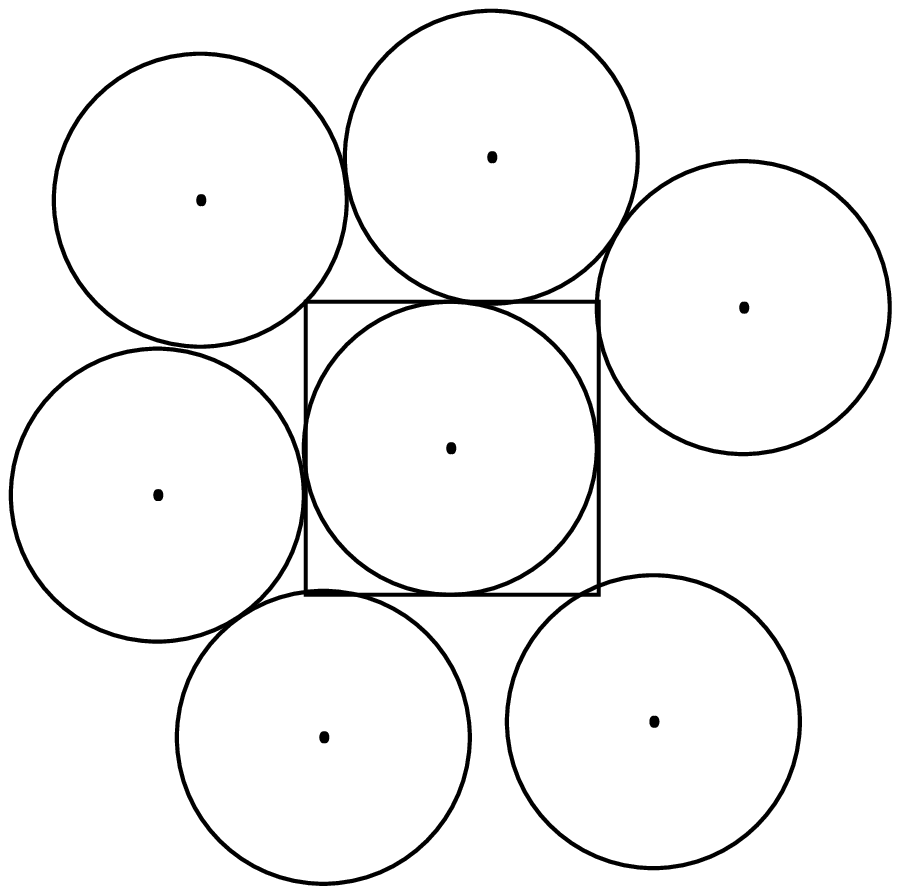}
}
\caption{Various possibilities by trial and error for even values of $\tau$.}
\label{1.observb}
\end{figure}

 \noindent{\bf ${\cal C}'$ with odd side length:} When {$\tau$} is not a multiple of 2, again it is quite intutive that a symmetric pattern is likely to give the best results. So we attempt two types of symmetric pattern.
 
 \textbf{Case 1:} Like the case for even values of $\tau$ we consider the symmetric distribution along the horizontal and vertical axes. Again we have two possibilities either the center of disk lying along the axes or symmetric about the axes for both the pairs of edges. It can be shown that for both the pair of edges we would have only one amongst the two configurations at a time for getting the miximum number of disks. In the first case, $\lceil \frac{\tau}{2}\rceil$ disks can completely be accomodated on one edge and one disk as a common disk between two adjacent edges. So, there will be $4\times \lceil \frac{\tau}{2}\rceil+4$ disks in the exterior part for this case. The number of disks in the interior part as in the even case would be $(\frac{\tau}{2})^2$. So, the total becomes $4\times\lceil\frac{\tau}{2}\rceil+4+(\frac{\tau}{2})^2$ (see Fig. \ref{bound2a} and Fig. \ref{bound2b}).
 
 \textbf{Case 2:} When diagonally symmetric, the case is quite similar to the previous case with 4 disks at the corners and $\lceil\frac{\tau}{2}\rceil$ among the four edges of ${\cal C}'$. The total again is calculating to $4\times\lceil\frac{\tau}{2}\rceil+4+(\frac{\tau}{2})^2$. 

 \begin{figure}
\centering
\subfloat[Center of disks along the axes\label{bound2a}]{
\includegraphics[width=.45\textwidth]{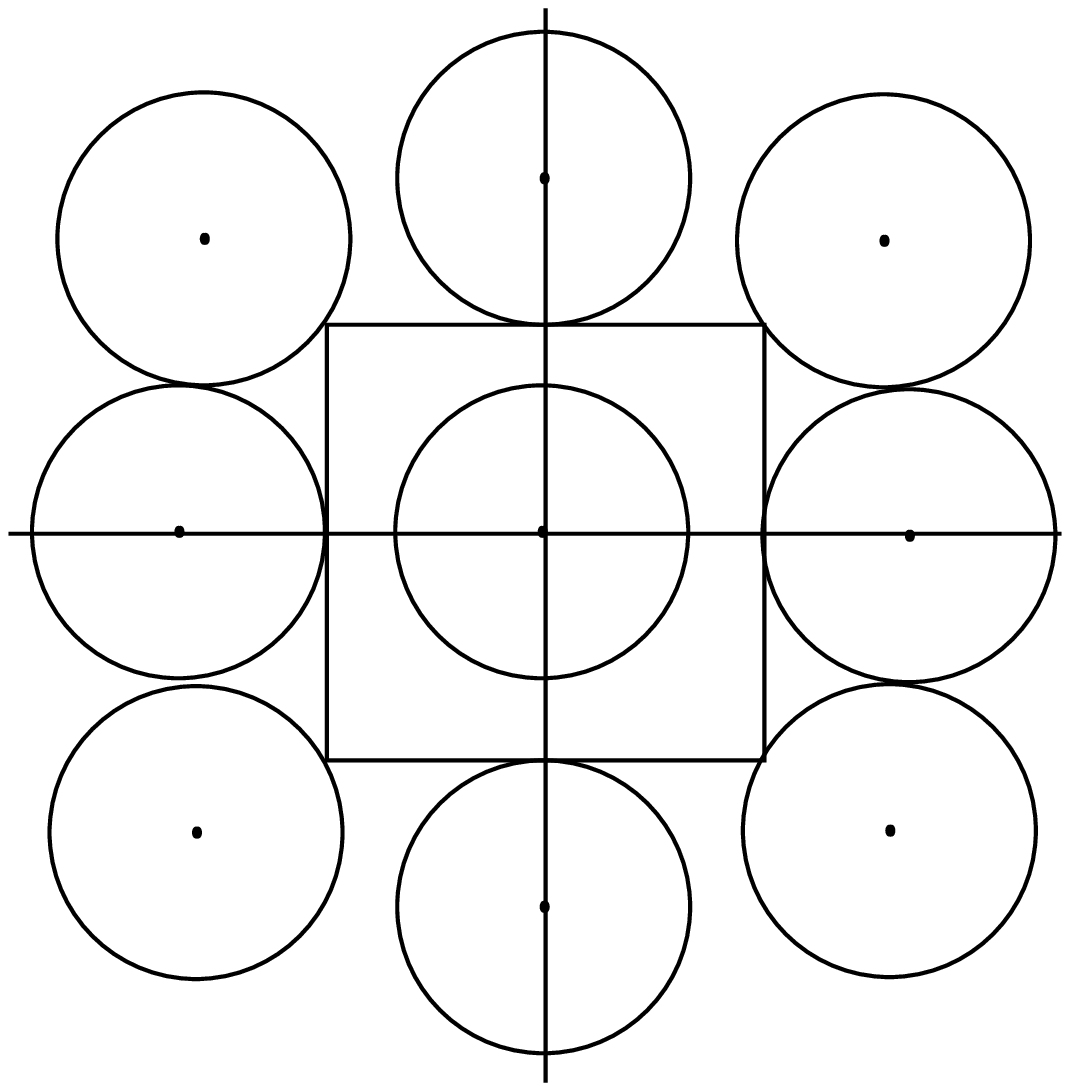}
}
\subfloat[Axes as tangent to the disks\label{bound2b}]{
\includegraphics[width=.45\textwidth]{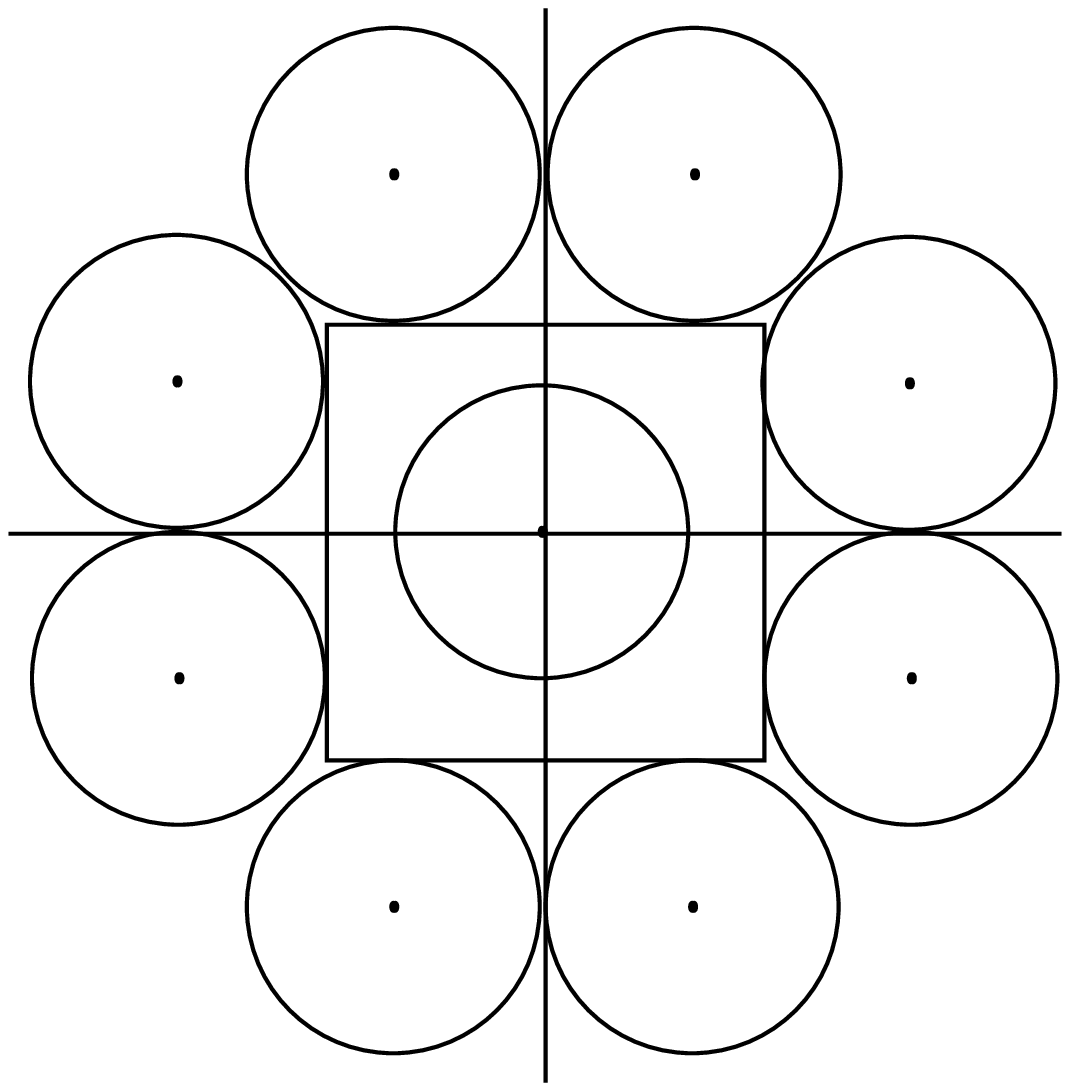}
}

\caption{Possible symmetric arrangement for odd values of $\tau$.}
\label{1.observb}
\end{figure}

 
\end{proof}

\begin{lemma} \label{lemma-1.1}
If ${\cal C}'$ is a grid cell of size $2\times 2$ and $D_{{\cal C}'}\subseteq D$ is a $k$-colorable solution for $P\cap {\cal C}'$, then $|D_{{\cal C}'}|\leq 7k$. 
\end{lemma}

\begin{proof}
 From Observation \ref{observation1.1} the cardinality of any set $S_{{\cal C}'}$ of pairwise disjoint unit disks intersecting with a grid cell ${\cal C}'$ is at most 7 as $\tau=2$.  Therefore, $|S_{C'}|\leq 7$. Now consider another $k-1$ sets $S_{{\cal C}'}^1$, $S_{{\cal C}'}^2$, $\ldots$, $S_{{\cal C}'}^{k-1}$, each of which can either be a replica of the same collection of disks in $S_{{\cal C}'}$ or a rotation or transformation of $S_{{\cal C}'}$ such that disks within each set remain pairwise disjoint and intersect ${\cal C}'$. Hence, any $k$-colorable solution $D_{{\cal C}'}\subseteq S_{{\cal C}'}\cup S_{{\cal C}'}^1\cup S_{{\cal C}'}^2 \cup \ldots \cup S_{{\cal C}'}^{k-1}$ if the union of the disks in $\big(\bigcup\limits_{i=1}^{k-1} S_{{\cal C}'}^{i}\big)\cup S_{{\cal C}'}$ covers all the points in $P\cap {\cal C}'$ .  Thus, the lemma follows. 
\end{proof}


The outline of our algorithm (Algorithm \ref{4.alg:SSC}) for computing a cover $D'\subseteq D$ of the points $P$ is as follows. We first partition the rectangular region containing the objects in $D$ and $P$ into individual grid cells of size $\tau \times \tau$. By utilizing the bound obtained in Observation \ref{observation1.1} (for e.g., for $\tau=2$ the actual bound is in Lemma \ref{lemma-1.1}) we compute a $k$-colorable cover of the points lying in each grid cell, in an exhaustive manner. To ensure that there is no conflict in the overall aggregate solution, we use a handing-over logic. Only disks of any particular grid cell cover centered within the same grid cell are colored with the associated color set. If a disk is required to be a part of this grid cell cover, but is centered in another grid cell, it is handed-over to that grid cell. Based on the grid width $\tau$ and the diameter of the disk, we then define a coloring scheme $\chi$ that assigns a color to each disk in the union $D'$ of all the individual grid cell covers computed. Finally, we return the pair $(D',\chi)$. Since the diameter of the disks is fixed to be two units, the approximation factor of the algorithm is implied by the choice of the value $\tau$. If the value of $\tau$ is 2, then a unit disk can participate in the $k$-colorable covers of points lying in four adjacent grid cells. Hence, we prove that Algorithm \ref{4.alg:SSC} is a 4-approximate algorithm (see Theorem \ref{theorem-2.2}). Later, we show that by varying the grid width $\tau$, which results in a unit disk participating in more than four individual grid cell covers, we can obtain a family of algorithms with  approximation factors corresponding to the choice of the value of $\tau$ (see Subsection \ref{generalize2.2}).

We now define any coloring function that assigns colors to disks to be conflict-free if for any pair of non-disjoint disks (i.e., overlapping disks) the colors assigned to them are different.

\begin{lemma} \label{lemma-conflict}
 The coloring $\chi$ defined by Algorithm \ref{4.alg:SSC} is conflict-free.
\end{lemma}
\begin{proof}

 For the sake of contradiction, let us assume that there are two disks 
 $d, d' \in D'$ such that $d\cap d'\neq\emptyset$, and $\chi(d)=\chi(d')$, where $D'$ along with $\chi$ is the output of Algorithm \ref{4.alg:SSC}. Since $d\cap d'\neq\emptyset$, the distance between the centers of $d$ and $d'$ is at most 2. Let the centers of $d$ and $d'$ be lying in the grid cells ${\cal C}$ and ${\cal C}'$, respectively. Observe that ${\cal C}$ and ${\cal C}'$ are either linearly or diagonally adjacent. If $d$ and $d'$ are chosen to cover points lying only in the respective grid cells, then $d$ and $d'$ are assigned colors from different color sets because the row and column numbers\hspace{-0.3em}$\mod 2$ in their ID's are not the same for both (see for-loop at Line 17) (contradicting that $\chi(d)=\chi(d')$). Therefore, the only possibility for color-conflict to arise between $d$ and $d'$ is that when both $d$ and $d'$ are centered in the same grid cell ${\cal C}$, where $d$ covers a point lying in the cell above ${\cal C}$ and $d'$ covers a point lying in the cell below ${\cal C}$ and each disk is initially chosen by the respective grid cell by means of the algorithm (Note that a similar case can be studied for horizontally and diagonally opposite grid cells). As per our color scheme (Line 16-28), these grid cells are the nearest to have the same color set (say $C_1$) associated with them (see Fig. \ref{fig-2}). Step 11 in the algorithm solves the conflict that arises in this case as follows. By means of the grid cell ID condition, disk $d'$ is handed over to grid cell ${\cal C}$ as its ID is greater and will receive the color set associated with that cell ($C_3$ in this case, see Fig. \ref{fig-2}). Disk $d$ however will not be handed over, but retains a color from the color set $C_1$ (contradicting that $\chi(d)=\chi(d')$). Thus, the lemma follows.  
\end{proof}

\begin{algorithm}
\caption{K\_Colorable\_Cover$(P, D, k)$}

\hspace*{\algorithmicindent} \textbf{Input:} A set $P$ of $n$ points, a set $D$ of $m$ unit disks in the plane, and an integer $k (>0)$ such that $P\subset\cup_{d\in D}d$ and $D$ can provide a $k$-colorable cover of the points in $P$. 
\\
\hspace*{\algorithmicindent} \textbf{Output:} A $k$-colorable set $D'\subseteq D$ that covers all the points in ${\cal P}$ and a color mapping $\chi:D' \rightarrow \kappa$, where $\kappa$ denotes the color set of distinct colors, and $|\kappa|\leq 4k$
\begin{algorithmic}[1]
\State Let the points in $P$ and disks in $D$ be lying entirely within the first quadrant of the coordinate system, and ${\cal R}$ be an axis-aligned rectangular region containing $P$ and $D$, whose left and bottom boundary lines coincide with the $y$- and $x$-axes of the coordinate system, respectively.
\State Define a grid ${\cal G}$ that partitions ${\cal R}$ such that each grid cell is of size $2 \times 2$ and for each point $p=(x_p, y_p)$ lying in such a cell ${\cal C}$, let the unique id associated with ${\cal C}$ be $id_{\cal C}=[\lfloor\frac{x_p}{2}\rfloor, \lfloor\frac{y_p}{2}\rfloor]$. For each such cell ${\cal C}$, we also define a handover set $H_{\cal C} \leftarrow \emptyset$.\newline
/* the grid cells in the following loop are considered in row-wise order from bottom-left to top-right, as defined in Subsection \ref{section2.1} */
\For {each grid cell ${\cal C}$ if $P\cap {\cal C}\neq\emptyset$}
\If {$H_{\cal C} = \emptyset$}
\State Let $D''=\{d\in D \hspace{2mm}|\hspace{2mm} d\cap {\cal C}\neq\emptyset \}$
\State Generate all subsets $D_1, D_2, \ldots, D_{O(m^7)}\subseteq D''$, each containing at most 7 pairwise disjoint disks, and among these, choose $k$ subsets $ S_{{\cal C}}^1,  S_{{\cal C}}^2, \ldots, S_{{\cal C}}^{k}$, whose union covers all the points in $P\cap {\cal C}$. 
\Else 
\State Let $D''=\{d\in D \hspace{2mm}|\hspace{2mm} d\cap {\cal C}\neq\emptyset,  d\notin H_{\cal C} \}$. 
\State Generate all subsets $D_1, D_2, \ldots, D_{O(m^7)}\subseteq D''$, each containing at most 7 pairwise disjoint disks, and among these, choose $k$ subsets $ S_{{\cal C}}^1,  S_{{\cal C}}^2, \ldots, S_{{\cal C}}^{k}$, whose union covers all the points in $P\cap {\cal C}$ but also contains all disks $d \in H_{\cal C}$. \newline /* $|S_{{\cal C}}^1\cup S_{{\cal C}}^2\cup \ldots S_{{\cal C}}^{k}\cup  H_{\cal C}|\leq 7k$ due to Lemma \ref{lemma-1.1} */
\EndIf
\State If any disk $d$ in any subset $S_{{\cal C}}^{i}$ (for $i=1, \ldots, k$) is centered in another grid cell ${\cal C'}$ whose ID $id_{{\cal C}'} > id_{\cal C}$, we remove that disk from $S_{\cal C}^{i}$ and add it to the handover set of that cell $H_{{\cal C}'}$.
\State $D_{\cal C}\leftarrow S_{{\cal C}}^1\cup S_{{\cal C}}^2 \cup \ldots \cup S_{{\cal C}}^{k}$
\State For every point $p \in P$ that is covered by a disk $d \in D_{\cal C}$ we remove it from $P$.

\EndFor
\algstore{myalg}
\end{algorithmic}
\label{4.alg:SSC}
\end{algorithm}

\begin{algorithm}
\begin{algorithmic}[1]
\algrestore{myalg}
\State $D'\leftarrow \bigcup\limits_{{\cal C}, P\cap{\cal C}\neq \emptyset}D_{\cal C}$
\State Let $C_1, C_2, C_3, C_4$ be four disjoint color sets, each containing $k$ distinct colours.
\For {every grid cell ${\cal C}$ with $id_c = [i, j]$}
\If {$((i\bmod2=0) \land (j\bmod2=0))$}
    \State Assign $C_1$ to ${\cal C}$.
\ElsIf {$((i\bmod2=0) \land (j\bmod2\neq0))$}
    \State Assign $C_2$ to ${\cal C}$.
\ElsIf {$((i\bmod2\neq0) \land (j\bmod2=0))$}
    \State Assign $C_3$ to ${\cal C}$.
\ElsIf {$((i\bmod2\neq0) \land (j\bmod2\neq0))$}
    \State Assign $C_4$ to ${\cal C}$.
\EndIf
\EndFor
\State For every grid cell ${\cal C}$ and its assigned color set $C_i$ we allot one color each to the subsets $S_C^1, S_C^2, \ldots S_C^k$. Every disk centered within that subset will now be colored with the corresponding color from the color set. 
\State For any disk $d\in D'$ let $\chi(d)$ represents the color assigned to the disk $d$ in the above coloring assignment process.\\
\Return{$(D', \chi)$}
\end{algorithmic}
\end{algorithm}


%

\begin{theorem} \label{theorem-2.2}
Algorithm \ref{4.alg:SSC} is a 4-approximation algorithm that runs in $O(m^{7k}n\log k)$ time for the $k$-CDUDC problem.
\end{theorem}
\begin{proof}
In line 2 of Algorithm \ref{4.alg:SSC}, the rectangular region ${\cal R}$ is partitioned into $O(n)$ grid cells. For each grid cell ${\cal C}$, let $n_{\cal C}$ denote the number of points of $P$ lying in ${\cal C}$, i.e., $n_{\cal C}=|P\cap {\cal C}|$. For each grid cell ${\cal C}$ we then enumerate all $O(m^7)$ subsets of $D$ such that each such subset contains at most 7 pairwise disjoint disks intersecting with the cell ${\cal C}$. In each iteration of the for-loop at line 3, among ${cm^7\choose k }$ collections of $k$ subsets together containing at most $7k$ disks for some constant $c$, we compute a set $D_{\cal C}$ that covers all $n_{\cal C}$ points lying in ${\cal C}$. In order to do this, we first compute the voronoi diagram $VOD_{\cal C}$ on the center points of these $7k$ disks. We then do point location queries for each of these $n_{\cal C}$ points on $VOD_{\cal C}$ to determine the closest center point. We will then test whether the corresponding disk centered at the closest center point covers this point. This step will take $O(k\log k+n_{\cal C}\log k)$ time. Therefore, we invest $O(m^{7k}(k\log k+n_{\cal C}\log k))$ time to compute $k$-colorable cover of the points lying in the cell ${\cal C}$. The total time the for-loop takes over all the nonempty grid cells is $O(m^{7k}k\log k)+ \sum\limits_{{\cal C}}O(m^{7k}n_{\cal C}\log k)$=$O(m^{7k}(n+k)\log k)$. The remaining steps of the algorithm (including the preprocessing in steps 1 and 2) will take no more than $O(mn)$ time. Under the assumption that $n$ is much larger than the number $k$ of colors, the time complexity of the algorithm in total is $O(m^{7k}n\log k)$ . In order to assign the colors at a later step of the algorithm, with each grid cell ${\cal C}$ we also associate and store the corresponding cover $D_{\cal C}$ computed in the for-loop at line 3. Thus, the additional space the algorithm requires is $O(nk)$.  
 
The handover behaviour at line 11 of the algorithm and Lemma \ref{lemma-conflict} ensure that any color-conflict is resolved in an elegant way. Since we use four disjoint sets of $k$ distinct colors and for each grid cell we compute $k$-colorable unit disk cover (from Lemma \ref{lemma-1.1}), the approximation factor of the algorithm is 4. Thus, the theorem follows.   
 
\end{proof}

The following corollary says that Theorem \ref{theorem-2.2} yields a faster algorithm for {\it $3$-CDUDC} than that of Biedl et al.\cite{BBL19}, but at the cost of increase in approximation factor.
\begin{corollary} \label{corollary-4.3}
There exists a 4-approximate algorithm that runs in $O(m^{21}n)$ time for the {\it $3$-CDUDC} problem.
\end{corollary}
\begin{proof}
 Follows from Theorem \ref{theorem-2.2} by applying $k=3$.
\end{proof}

\subsection{Improved Algorithm for a Grid Cell}\label{localimproved2.4}
 In this subsection, we try to improve on the brute force approach of selecting 7 or fewer disks (7 being the max when $\tau$ is 2). We suggest various possible ways of improvement.
  
We divide our grid cell $\it C$ of size $\tau\times \tau$ into 4 equal subgrid cells (see Fig. \ref{1.improv1}(a)). Therefore, the edge of each subgrid cell would be $\frac{\tau}{2}$, i.e., it would be between $\frac{1}{2}$ to 1, since $1 \leq \tau \leq 2$. We consider a subset $D'$ of disks that can cover points in 2 diagonally opposite pair of subgrid cells such that the center of these disks lies in the respective quadrant of the subgrid cells, assuming that $\it C$ is placed such that the center of $\it C$ is at the origin of the $x, y$-coordinate system. Suppose that ${\it C}_1$ and ${\it C}_2$ are these subgrid cells chosen and that $D'$ would contain only those disks that have centers above and towards the right for ${\it C}_1$, and below and towards left for the dividing boundaries for ${\it C}_2$ (see Fig. \ref{1.improv1}). 
 
 \begin{observation}
  The maximum number of pairwise disjoint disks in $D'$ is at most 5.
 \end{observation}

\begin{proof}
  This can be estimated by considering our earlier observation that for a grid cell of edge length less than 1, at most 4 disjoint disks are required to cover it maximally. The 4 disks required can be seen as 1 disk for covering each edge. Effectively when we consider disks with centers in a particular quadrant, we are covering only 3 sides for each subgrid cell. The sub grid cells being diagonally opposite have a common disk, so effectively we are looking for $3 + 3 - 1 = 5$ non-intersecting disks to cover both the diagonally opposite subgrid cells maximally 
  (see Fig. \ref{1.improv1}(b)). Also this being an upper bound for case where $\tau$ is 2 and reduces as the value of $\tau$ reduces. 
\end{proof}

Given the points lying in $(C_1\cup C_2)$ and being covered by $D'$, it is a fair possibility that we might require less than 5 disks to cover the points in diagonally opposite sub-grid cells $C_1$ and $C_2$. So we calculate combinations of $1, 2, \ldots, 5$ disks of $D'$ such that they cover the points in diagonally opposite subcells and are pairwise interior-disjoint. For each of this combination calculations we require ${m\choose 1}, {m\choose 2}, \ldots , {m\choose 5}$ possibilities respectively making the complexity to be $O(m^{5})$. For each of these possibilities, we remove from the given disk set $D$ all those disks that might overlap with the chosen disks and then arrange the remaining disk centers lying in the quadrants containing the remaining subgrid cells in increasing values of $x$ and $y$ coordinate. We choose the disk for the upper leftmost subgrid cell (call it $C_3$) with its center having minimum $x$ and $y$ coordinates (i.e., a disk whose center is lexicographically closest to $C_3$) covering all the required points. Likewise, we choose the disk for the bottom rightmost grid cell (call it $C_4$) with its center having a minimum $x$ and $y$ coordinate covering the points. Overall the complexity for sorting turns out to be $O(m\log m)$. Hence the overall complexity for selecting the disks reduces from $O(m^{7})$ (ignoring $n$ and $k$) to $O(m^{6} \log m)$ for $m$ disks for a given color.

 \begin{figure}[h]
\centering
\subfloat[Dividing grid cell ${\it C}$ into 4 smaller subgrid cells]{
\includegraphics[width=.4\textwidth]{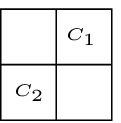}
}
\subfloat[Maximum possible number of disks in $D'$]{
\includegraphics[width=.6\textwidth]{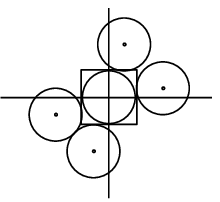}
}
\caption{Improvement in running time}
\label{1.improv1}
\end{figure}

We also repeat the above procedure by considering all possible combinations ${m\choose 1}, {m\choose 2}, \ldots , {m\choose 5}$ for the diagonally opposite subgrid cells $C_3$ and $C_4$ and for each subset of at most 5 non-intersecting disks from each of such combinations by finding two disks (pairwise disjoint from each of these 5 disks), lexicographically closest to subgrid cells $C_1$ and $C_2$ respectively, as above. The overall time is still $O(m^{6} \log m)$ excluding the factors $n$ and $k$.

However, by presorting the disks in $D$ by both the $x$- and $y$-coordinates of their centers and by maintaining these two orderings separately, we can improve the running time from $O(m^{6} \log m)$ to $O(m^{6})$ because given the cell $C$, extracting the sorted lists of disks for the remaining subgrid cells and choosing the disks from these lists to cover the required remaining points all can be done in $O(m)$ time. 

Hence, for a grid cell ${\cal C}$ we can generate all candidate subsets of at most 7 pairwise disjoint disks corresponding to each of $k$ colors, in $O(m^6)$ time. Thus, the running time of our exaustive search for a $k$-colorable cover of all points lying in a cell ${\cal C}$ can be improved by a factor of $O(m)$. The correctness of this procedure follows again due to the packing constraints (Lemma \ref{lemma-1.1}). Therefore, we have the following results.

\begin{theorem} \label{theorem-2.2.1}
We have a 4-approximation algorithm that runs in $O(m^{6k}n\log k)$ time for the $k$-CDUDC problem.
\end{theorem}

\begin{corollary} \label{corollary-4.3}
There exists a 4-approximate algorithm that runs in $O(m^{18}n)$ time for the {\it $3$-CDUDC} problem.
\end{corollary}

\subsection{Generalization}\label{generalize2.2}
In this subsection, we generalize the results from the preceding subsection to a general case observations for possible values of the width $\tau$ of each cell in the grid partitioning approach. We begin by attempting to generalize the potential coloring schemes for each grid cell $\tau\times \tau$ similar to the observations presented in Lemma \ref{lemma-1.1}. We first define a parameter $\rho$ that represents the factor indicating the number of additional color sets needed to satisfy a union of independent solution sets. As a result, we obtain a family of approximation algorithms for the $k$-{\it CDUDC} problem, depending upon a different possible values for $\rho$ and $\tau$. 

\begin{lemma} \label{lemma4}
    The union of all independently optimal $k$-colorable solution sets for points lying in each grid cell ${\cal C}$ of size $\tau\times\tau$ is $\rho k$ colorable, where
    
    $\rho = \begin{dcases*}
    4 & \hspace{8mm}if $\tau\geq2 $\\
    6 & \hspace{8mm}if $\frac{8}{5}\leq \tau < 2$\\
    7 & \hspace{8mm}if $\sqrt{2}\leq \tau < \frac{8}{5}$\\
    9 & \hspace{8mm}if $1 \leq \tau < \sqrt{2}$
    \end{dcases*}$
 \end{lemma} 
 \begin{proof}
    We prove this by verifying the number of grid cells a unit disk can maximally intersect while considering each case. This determines the number of disjoint color sets, each consisting of at most $k$ distinct colors. Clearly, this number is the same as $\rho$. Note that we are interested in finding the upper bound of such intersections.
    \begin{figure}[h]
\centering
\begin{center}
\includegraphics[scale=0.9, angle =-90]{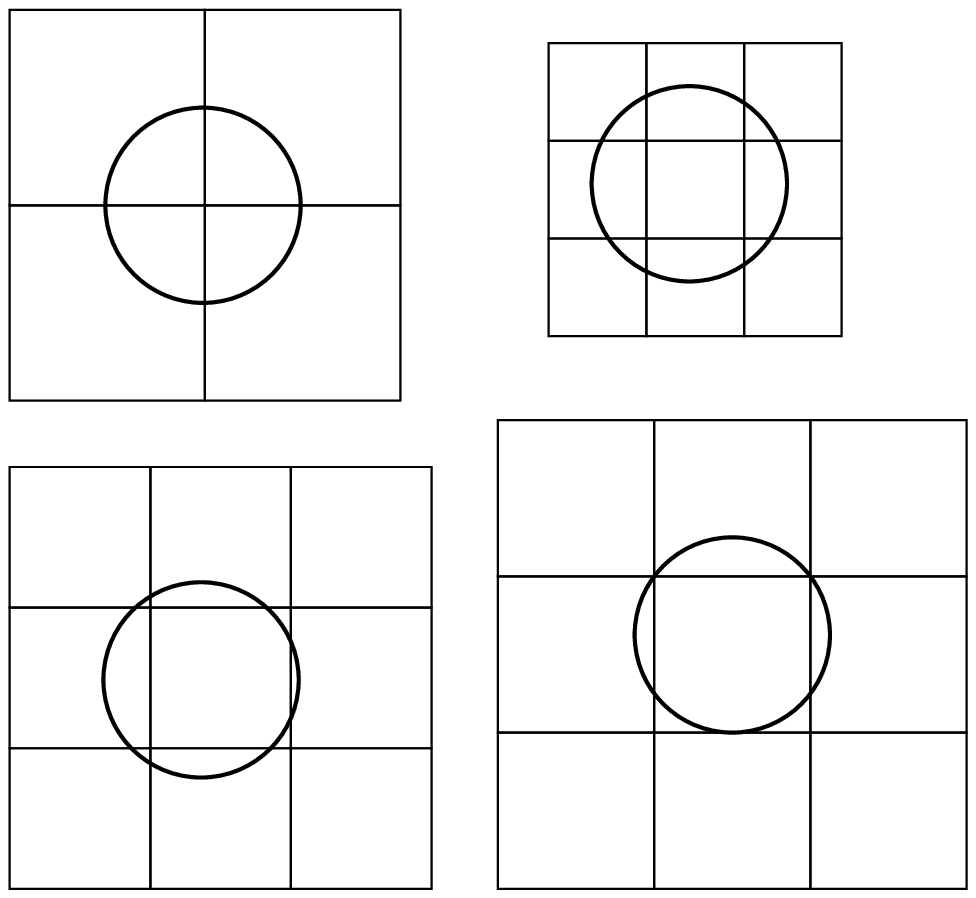}
\end{center}
\begin{picture}(0,0)
\put(  -75, 0){$(c)$}
\put(  68, 0){$(d)$}
\put(  -75, 128){$(a)$}
\put(  58, 145){$(b)$}
\end{picture}
\caption {Proof of Lemma \ref{lemma4}}
\label{fig-lemma4}
\end{figure}

\begin{mycases}
 \case $\tau\geq2 $: Since the diameter of each disk is 2 and the width of grid cell is also ($\geq$)2, certainly, no disk can span more than 2 linearly adjacent disks (see Fig. \ref{fig-lemma4}$(a)$). Thus, maximal intersection count is achieved by placing the disk in any of the grid intersection corners. Here the disk will certainly intersect $4$ grid cells regardless of the width. 
 
 \case $\frac{8}{5}\leq \tau < 2$ : If the width of the grid cell is less than 2, surely, a unit disk can span 3 linearly adjacent grid cells. However, if the width is $\frac{8}{5}$, the disk cannot span more than 2 diagonally adjacent grid cells (see Fig. \ref{fig-lemma4}$(d)$). Hence, if the width is greater than or equal to $\frac{8}{5}$, then no matter where the disk is centered (in Fig. \ref{fig-lemma4}$(d)$), it can not intersect more than 6 grid cells simultaneously.

 \case $\sqrt{2}\leq \tau < \frac{8}{5}$ : If the width of the grid cell is $\sqrt{2}$, surely, a unit disk can span 3 linearly adjacent grid cells and the middle cell from the next adjacent column or row of the linearly adjacent grid cells (see Fig. \ref{fig-lemma4}$(c)$).  Thus, we use the position shown in Fig. \ref{fig-lemma4}$(c)$ to indicate maximal count possible. 

 \case $1 \leq \tau < \sqrt{2}$ : For width $\geq 1$, a unit disk could potentially intersect 3 linearly and diagonally adjacent grid cells (see Fig. \ref{fig-lemma4}$(b)$). 
 \end{mycases}


 We do not consider the grid partitioning with grid cells of width $\tau < 1$ as it is inconsequential to our study and provides no substantial results (as the colorability increases substantially with no real improvement to the number of intersecting disks).
\end{proof}

 In the same spirit, we provide a generalization of the observation presented in Lemma \ref{lemma4} for the maximum number of pairwise disjoint unit disks that can intersect a square of size $\tau\times \tau$. We define a parameter $\alpha$ that is the count of mutually non-intersecting unit disks that can intersect a grid cell of width $\tau$.

 \begin{lemma}\label{lemma5}
    The bound on the cardinality of a set $S_{\cal C}^i$ of pairwise disjoint unit disks that can intersect a grid cell ${\cal C}$ of size  $\tau\times\tau$ is $\alpha$, for an integer $1\leq i \leq k$, where
    
    $\alpha = \begin{dcases*}
    4 & \hspace{8mm}if $\tau = 1$\\
    5 & \hspace{8mm}if $\frac{8}{5} \geq \tau \geq \sqrt{2}$\\
    7 & \hspace{8mm}if $\tau = 2$\\
    10 & \hspace{8mm}if $\tau = 3$\\
    14 & \hspace{8mm}if $\tau = 4$\\
    17 & \hspace{8mm}if $\tau = 5$
    \end{dcases*}$
  \end{lemma}  
 \begin{proof}
 
  Consider a grid cell ${\cal C}$ of size $1\times 1$. Imagine placing four unit disks, each centered farthest apart from one another, but outside the cell ${\cal C}$, and touching one of the four corners of the cell ${\cal C}$. Since the cell size is $1\times 1$, we can not place any more disk that intersects the cell, but at the same time disjoint from each of these four disks. Hence, the bound $\alpha=4$ for the case of width $\tau = 1$. Proof for each of the other cases can be done similar to the proofs provided for the cases $\tau=1$ given above, and $\tau=2$ given in Lemma \ref{lemma-1.1} (see, for example, the cases $\tau = 1$ and $\tau=8/5$ being illustrated in Fig. \ref{fig-lemma5}).
   \begin{figure}[h]
\centering
\begin{center}
\includegraphics[scale=0.9, angle =0]{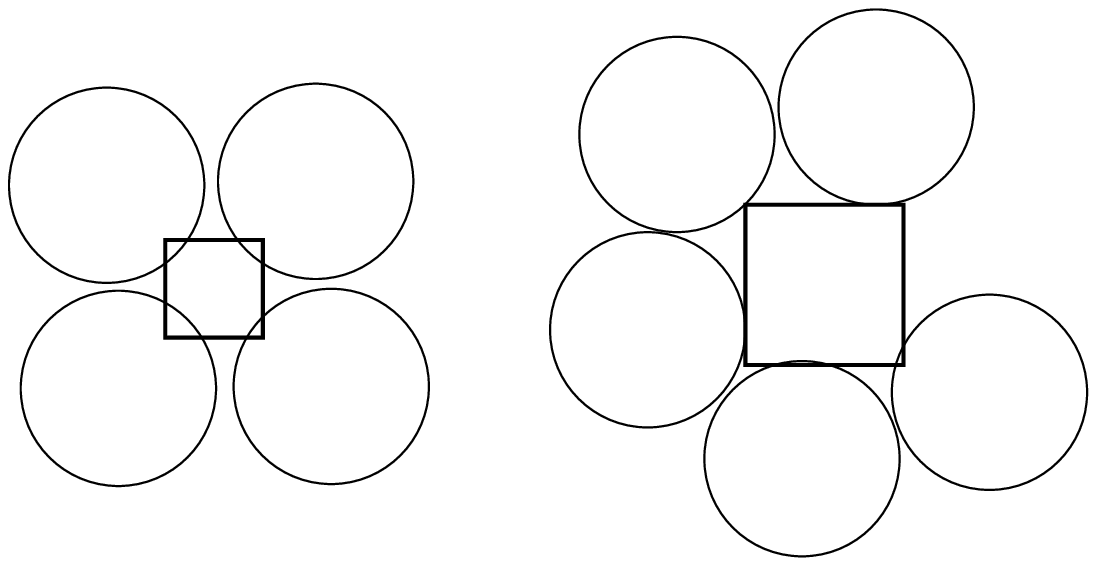}
\end{center}
\begin{picture}(0,0)
\put(  -100, 10){$(a)$ $\tau=1$}
\put(  58, 10){$(b)$ $\tau=\frac{8}{5}$}
\end{picture}
\caption {Proof of Lemma \ref{lemma5}}
\label{fig-lemma5}
\end{figure}
  We do not consider grid cells of width $\tau>5$ as it provides no improvement to the running time while the approximation factor stays the same after $\tau=2$ (as seen in Lemma \ref{lemma4}).  On the other hand, $\tau< 1$ is also not considered because the approximation factor $\rho$ becomes arbitrarily very high for $\tau< 1$. One can observe here that the upper bound for $|S_{\cal C}^i|$ obtained in Observation \ref{observation1.1} is tight in case of $\tau$ being even, whereas in case of odd, it is a loosely bound.
  \end{proof}


 \begin{theorem} \label{theorem-4}
There exists a $\rho$-approximation algorithm to solve the $k$-{\it CDUDC} problem, that has a runing time of $O(m^{\alpha k}n\log k)$ for a given grid width $1 \leq \tau \leq 5$, where

    $\rho = \begin{dcases*}
    4 & \hspace{2mm}if $\tau\geq2$ \\
    6 & \hspace{2mm}if $\frac{8}{5}\leq \tau < 2$\\
    7 & \hspace{2mm}if $\sqrt{2}\leq \tau < \frac{8}{5}$\\
    9 & \hspace{2mm}if $1 \leq \tau < \sqrt{2}$
    \end{dcases*}$ 
    , \quad
    $\alpha = \begin{dcases*}
    4 & \hspace{2mm}if $\tau = 1$\\
    5 & \hspace{2mm}if $\frac{8}{5} \geq \tau \geq \sqrt{2}$\\
    7 & \hspace{2mm}if $\tau = 2$\\
    10 & \hspace{2mm}if $\tau = 3$\\
    14 & \hspace{2mm}if $\tau = 4$\\
    17 & \hspace{2mm}if $\tau = 5$
    \end{dcases*}$    
    
\end{theorem}
\begin{proof}
The input of Algorithm \ref{4.alg:SSC}, in addition to a set $P$ of $n$ points, a set $D$ of $m$ unit disks, and an integer $k (>0)$, also consists of a grid partitioning parameter $\tau$. From Lemmata \ref{lemma4} and \ref{lemma5}, it is clear that any reasonable value for the parameter $\tau$ will imply the values of $\rho$ and $\alpha$. Hence,  we have a $\rho$-approximation algorithm in $O(m^{\alpha k}n\log k)$ time. 
\end{proof}
 

 \subsection{Further Generalization}\label{fgeneralize2.3}
 
 In this subsection, as in Theorem \ref{theorem-4} we further attempt to generalize the packing constraints that define the values of $\rho$ and $\alpha$ our stated algorithm might take for a given value of $\tau$. We attempt to bring an upper bound to the values of $\alpha$ and $\rho$ and hence an approximate algorithm with complexity in terms of $\tau$.
 
 The density of a packing of two or more objects in the interior of any region $R$ is the ratio between the area of the union of these objects and the total area of $R$. Then, we have the following packing lemmas.
 
 \begin{lemma}[Fejes T{\'o}th, \cite{FT53}] 
    \label{packing-lemma-1}
    Every packing of two or more congruent disks in a convex region has density at most $\frac{\pi}{\sqrt{12}}$.
\end{lemma}
 
  \begin{lemma}
    \label{packing-lemma-2}
    Consider a unit ball $b$ in the space $\mathbb{R}^d$, and consider any collection $X$ of pairwise disjoint hydercubes of side lendths at least $\tau$ that overlap $b$, then 
  $|X|\leq  \left(1+\lceil\frac{2}{\tau}\rceil\right)^d$.
\end{lemma}

  \begin{theorem} \label{theorem-5}
There exists a $\rho$-approximation algorithm to solve the $k$-{\it CDUDC} problem, that has a runing time of $O(m^{\alpha k}n\log k)$ for a given grid width $1 \leq \tau \leq 2$, where $\rho=\left(1+\lceil\frac{2}{\tau}\rceil\right)^2$ and $\alpha=\frac{4\pi+8\tau+\tau^2}{\sqrt{12}}$.
\end{theorem}

\begin{proof}

The region $R$ that encloses all pairwise disjoint unit disks that participate in $k$-colorable cover of points lying in a grid cell ${\cal C}$ of size $\tau \times \tau$ is the Minkowski sum of ${\cal C}$ and a disk of radius 2. The total area of this region $R$ is $\tau^2+4\pi+8\tau$. It is obvious that this region is convex. Now, using Lemma \ref{packing-lemma-1} we can bound the the maximum number of pairwise disjoint unit disks that can cover points in ${\cal C}$ as follows. The density of packing $\alpha$ unit disks in the interior of $R$ is $\frac{\alpha \pi}{\tau^2+4\pi+8\tau}\leq \frac{\pi}{\sqrt{12}}$.
Therefore, $\alpha$ as referred to in the previous sections can be said to be bounded by $O(\tau^{2})$. For getting a bound on $\rho$ we try to enclose as many squares as possible along the diameter vertically and horizontally. As many as $\frac{2}{\tau}$ + 1 squares can be arranged along with one of the diagonals. Therefore arranging the squares as grid by setting an upper value for number of squares in rows and columns, at most $(\frac{2}{\tau} + 1)^{2}$ which is $\frac{4}{\tau^{2}}$ + $\frac{4}{\tau}$ +1 (as in Lemma \ref{packing-lemma-2}). Since $\frac{4}{\tau}$ would be the dominating term here for values of $\tau \geq 1$, hence it can be said that $\rho$ is bounded by $O(\frac{1}{\tau})$. Substituting the bounds of $\alpha$ and $\tau$ in Theorem \ref{theorem-4}, the algorithm is a $O(\frac{1}{\tau})$-approximation algorithm with a running time of $O(m^{\alpha k}n\log k)$, where $\alpha=\frac{4\pi+8\tau+\tau^2}{\sqrt{12}}$.

Now we argue that $\tau$ should be chosen such that $1 \leq \tau \leq 2$. To provide an appropriate argument for our chosen values of $\tau$ we consider the following 2 cases 
\\
\textbf{Case 1:} Consider cases where $\tau < 1$. The Observation \ref{observation1.1} holds true only for integer values of $\tau$. So for values of $\tau < 1$ it can be observed that there would be at most 4 disjoint disks, by keeping each of the disks at the 4 vertices. Therefore the value of $\alpha$, in this case, is bounded by 4. Hence the running time remains $O(m^{4k}n\log k)$. The approximation factor increases arbitrarily. It turns out that a large number of color sets get used without any improvement in running time. So values of $\tau < 1$ do not provide us with much of an advantage.
\\
\textbf{Case 2:} Consider cases where $\tau > 2$. As it can be observed from Lemma \ref{lemma4}, the approximation factor $\rho$ remains 4 for values of $\tau$ greater than 2, however, the runtime increases a lot since $\alpha$ is bounded by $O(\tau^{2})$. Clearly, there is a huge overhead of running time with no improvement in the approximation factor.

We can infer that values of $\tau$ such that $1 \leq \tau \leq 2$ are preferred. 
\end{proof}

 \section{ Line Segment and Rectangular Region Cover} \label{section-LSDC}
In the same spirit as the $k$-{\it CDUDC} problem is considered due to its practical application in frequency/channel assignment in wireless networks, we also define two problems, that generalize the locations of potential wireless clients from discrete set of points to line segments and from discrete set of points to a continuous rectangular region, namely, the {\it $k$-Colorable Line Segment Disk Cover ($k$-CLSDC)} and {\it $k$-Colorable Rectangular Region Cover ($k$-CRRC)} problems, respectively.

We begin our approach using a fundamental combinatorial result involving unit disks that helps us to transform the above problems into our original $k$-{\it CDUDC} problem. Given a set $D$ of $m$ unit disks in the plane, a sector is the smallest region bordered by the boundary lines of disks and is covered by the same set of disks in $D$. Thus, the arrangement of all disks of $D$ subdivides the plane into many sectors. It is not hard to show that the worst-case complexity of the arrangement of any set of $m$ unit disks is quadratic, as stated below.

\begin{observation}[Funke et al. \cite{FKKLS07}] 
    \label{sector-lemma}
    The number of sectors created by intersection of $m$ unit disks in $D$ is $O(m^2)$.
\end{observation}

To develop approximation algorithms for the $k$-{\it CLSDC} problem, we transform
every instance of $k$-{\it CLSDC} problem into an instance of {$k$-{\it CDUDC} problem as follows. In an instance of $k$-{\it CLSDC} problem, we have a set $D$ of $m$ unit disks covering a finite union of $n$ line segments of arbitrary length with arbitrary orientation, and an integer $k$, The objective, here, is to compute a $k$-colorable cover of all the line segments. We split each of these line segments into slices such that each such slice lies within some sector.  Now, for each subset of slices lying within a single sector, we add one point into the same sector and remove all the slices. This collection of points is referred to as $P'$. Hence, from Observation \ref{sector-lemma} we have that $|P'| =  O(m^2)$. This can be taken as an instance of the $k$-{\it CDUDC} problem, where $P'$ is taken as the input set of points $P$. 

Similarly, we can do a similar transformation for the $k$-{\it CRRC} problem. Here, we have a set $D$ of $m$ unit disks covering a continuous rectangular region ${\cal R}$. Our objective is to compute a 
$k$-colorable set of units disks such that ${\cal R}$ is covered by the union of these disks. As above, we split ${\cal R}$ into $O(m^2)$ sectors, as induced by the union of disks in $D$, and add one point into each sector.

For the $k$-{\it CLSDC} problem, the construction of the set $P'$ of points can be done as follows. We first preprocess the given set $D$ of $m$ unit disks into any reasonable data structure,  e.g., a doubly connected edge list ({\it DCEL}), in $O(m^2)$ time \cite{VSDO00}. We can build the Voronoi diagram $VOR_{D}$ on the center points of disks in $D$ in $O(m\log m)$ time. We store cross pointers between Voronoi cells (or disks) in $VOR_{D}$ and the faces (or sectors) of {\it DCEL} that are contributed by the corresponding disks. For each of the given $n$ line segments, we do point location query on $VOR_{D}$ for the left endpoint of the line segment to determine the disk in which it lies. We then follow the cross pointer to access the sector that contains it. Subsequently, we traverse the adjacent sectors of this sector in {\it DCEL}. As we do, we add points into those sectors ( also, into $P'$) that covers a portion of the line segment and mark the corresponding faces (or sectors) as processed in {\it DCEL}. This step will take $O(n\log m+m^2)$ time. In the case of the $k$-{\it CRRC} problem, the construction of $P'$ takes $O(m^2)$ time as we have to test whether each sector is intersected by ${\cal R}$.

Therefore, we have the following results for the $k$-{\it CLSDC} and $k$-{\it CRRC} problems. 

  \begin{theorem} \label{theorem-6}
We have a $\rho$-approximation algorithm to solve the $k$-{\it CLSDC} problem, that has a runing time of $O(n\log m+m^{\alpha k+2}\log k)$ for a given grid width
$1 \leq \tau \leq 2$, where  $\rho$ and $\alpha$ are defined as in Theorem \ref{theorem-5}.

    
\end{theorem}
\begin{proof}
Follows from Theorem \ref{theorem-5} by applying $|P|=|P'|=O(m^2)$, where $O(n\log m)$ is due to the preprocessing of the input before we run the algorithm of Theorem \ref{theorem-4} .
\end{proof}
   \begin{theorem} \label{theorem-6}
We have a $\rho$-approximation algorithm to solve the $k$-{\it CRRC} problem, that has a runing time of $O(m^{\alpha k+2}\log k)$ for a given grid width $1 \leq \tau \leq 2$, where  $\rho$ and $\alpha$ are defined as in Theorem \ref{theorem-5}.
\end{theorem}
\begin{proof}
Follows from Theorem \ref{theorem-5} by applying $|P|=|P'|=O(m^2)$.
\end{proof}
\section{Conclusion}
In this paper, we have proposed constant-factor approximation algorithms for computing $k$-colorable unit disk covering of points, line segments, and a rectangular region. The algorithm is based on exhaustively searching an optimal $k$-colorable cover of points lying within a grid cell, which is a square of constant size. We believe that there seems to be no room for further improvement of approximation factor to smaller than 4 using the square grid approach. However, as future work, one could improve the almost brute-force algorithm's running time for a grid cell by exploiting various packing-constrained covering properties about grid square and unit disks. This will improve the running time of the overall algorithm.






\bibliographystyle{elsarticle-num-names}
\bibliography{sample.bib}







\end{document}